\documentclass[runningheads]{llncs}
\usepackage{graphicx,color,amssymb,comment,algorithm2e,amsmath,soul,tikz,hyperref,setspace,float}
\usetikzlibrary{positioning,arrows,chains,matrix,positioning,scopes}

\begin{document}
\title{Payment Network Design with Fees}
%
%
\author{Georgia Avarikioti
\and Gerrit Janssen \and
Yuyi Wang \and
Roger Wattenhofer
}
\authorrunning{G. Avarikioti, G. Janssen, Y. Wang, R. Wattenhofer.}
%
\institute{ETH Zurich, Switzerland\\
\email{\{zetavar,gjanssen,yuwang,wattenhofer\}@ethz.ch}\\}
\maketitle              
\begin{abstract}
Payment channels are the most prominent solution to the blockchain scalability problem. We introduce the problem of network design with fees for payment channels from the perspective of a Payment Service Provider (PSP). 
Given a set of transactions, we examine the optimal graph structure and fee assignment to maximize the PSP's profit. A customer prefers to route transactions through the PSP's network if the cheapest path from sender to receiver is financially interesting, i.e., if the path costs less than the blockchain fee. 
When the graph structure is a tree, and the PSP facilitates all transactions, the problem can be formulated as a linear program. 
For a path graph, we present a polynomial time algorithm to assign optimal fees. We also show that the star network, where the center is an additional node acting as an intermediary, is a near-optimal solution to the network design problem. 
\end{abstract}
\section{Introduction}
Scaling the transaction throughput  on blockchain systems, such as Bitcoin \cite{nakamoto2008bitcoin} and Ethereum \cite{ethereum}, is a fundamental problem and  an active research direction \cite{croman2016scaling}. Many solutions have been proposed, in particular sharding \cite{kokoris2017omniledger,luu2016sharding}, sidechains \cite{back2014sidechains} and channels \cite{DW2015channels,poon2015lightning,raiden2017,lind2016teechan}. Channels seem to be the most promising solution since they allow transactions to occur securely off-chain, and use the blockchain only for resolving disputes. 
%

We study the problem from the viewpoint of a Payment Service Provider (PSP). The PSP wants to establish an alternative payment network for customers to execute transactions. We assume  a PSP can open a channel between two parties without acting as an intermediate node; this can be done using three-party channels. The two parties and the PSP join a three-party channel funded only by the PSP who then loans money to the other parties. We assume that the PSP will eventually get his money back in fiat currency as he provides a service similar to credit cards (the risk lies to the PSP). Furthermore, the PSP signs each new state if and only if the fees have the correct value. This way he enforces the fee assignment on the channels.

Initially, a PSP will compete with the blockchain: customers only prefer the alternative network if the total fees cost less than the blockchain. 
We introduce the network design problem for the PSP, whose goal is to decide the graph structure and the fee assignments in order to maximize its profit.

Our contributions are as follows.
First, we provide a linear program formulation for the problem on trees when the PSP wants to facilitate all transactions, proving that this problem variation is in the complexity class P. Then, we show that the optimal fee assignment for any path has only 0/1 values on the fees, assuming $1$ is the cost of posting a transaction on the blockchain, and we present an efficient dynamic programming algorithm to compute the optimal fees. In addition, we prove that the star network is a near-optimal solution of the general network design problem, when we allow an additional node to be added as a payment hub and assume the optimal network is connected. This implies that a PSP can achieve almost maximum profit by creating a payment hub, the construction of which has already been studied in \cite{green2017bolt,heilman2017tumblebit}.
%
%
%
\section{Preliminaries and Notation}
In this section, we define the \emph{Channels Network Design with Fees (CNDF)} problem. We assume the PSP can renew the channels and change the network structure in specific epochs to avoid timing attacks; hence, we only consider a limited set of transactions corresponding to an epoch. Now, given a set of transactions between a fixed number of participants, we wish to create a payment network and assign fees to its channels to maximize the profit for the PSP. To formally define the problem, we introduce the following notation. 
%

We define a channel network as a graph $G = (V, E)$ with a set of vertices $V$ and a set of edges $E$. Each node $v \in V$ denotes one of $n$ participants wishing to use our network, hence $|V| = n$.
An edge $e \in E$ between two nodes $i$ and $j$ represents an open channel $C_{ij}$, with $|E| = m$. Thus, the set of edges $E$ represents the open channels of our network. 
For simplicity, we assume that all edges of the graph are undirected, as we assume the capacity of every channel to be infinite, in other words, the PSP has deep pockets and is able to fund channels with a significant amount of capital. Further, we define the cost of each edge in the network to be 1. This represents the cost of opening a channel by submitting a funding transaction to the blockchain as described in \cite{DW2015channels,poon2015lightning}.

Given a sequence of transactions for $n$ participants, we can define a transaction matrix $T \in \mathbb{N}^{n \times n}$. An entry $T[i,j]$ denotes the number of transactions from  $i$ to $j$ and back. Note that matrix $T$ is symmetric since the transactions' direction do not matter. If there are no transactions for a pair $(i,j)$ of nodes, then the corresponding matrix entry is $0$. Transactions where  sender and  receiver are identical are meaningless, thus the diagonal entries of the  matrix are $0$.

For each edge $e \in E$ we can assign a fee $f_e \in \mathbb{R}$. We require every fee to be non-negative. 
Moreover, we require the fees to be at most $1$, which is the cost of any transaction on the blockchain. Allowing the fee on an edge to be more than $1$ is equivalent to deleting this edge from the network, since the customers will always prefer to use the blockchain where the transaction fee is $1$. We denote by $f_E$ a fee assignment for the set of edges $E$.


To measure the value of a network we introduce a profit function. The profit of a payment network depends on the structure of the underlying graph, the fee assignments and the transactions carried out between participants in the network. 
Given a transaction matrix $T$, we define the profit of a graph $G = (V, E)$ as follows:
\begin{align*}
p(G, T, f_E)& =  -m + \sum_{i,j \in V} \sum_{e \in path(i,j)} f_e \cdot X_{ij} \cdot T[i,j], \\[1em] 
\textrm{where } X_{ij}& =
\begin{cases} 
1, & \textrm{if the participant chooses to use the network,} \\
  & \textrm{~~i.e. } \sum_{e \in path(i,j)} f_e \leq 1 \\
0, & \textrm{otherwise} 
\end{cases}
\end{align*}

where $path(i,j)$ denotes the set of edges of the shortest path (cheapest sum of fees on edges)  from sender $i$ to receiver $j$ in the graph $G$.


We include a pair of nodes $(i,j)$ in the profit calculation only if this sum of fees on the shortest path is at most $1$. Finally, we subtract from the profit the number of edges $m$, since each transaction that opens a channel costs $1$ in the blockchain.

Now, we formally define the problem as follows.
\begin{definition}(CNDF) \label{def: CNDF}
Given a transaction matrix $T \in \mathbb{N}^{n \times n}$, return a graph $G=(V,E)$ with $|V|=n$, and fee assignments on edges $f_E$, such that the profit function $p(G, T, f_E)$ is maximized. 
\end{definition}
In the following two sections 
we study a relaxed version of the problem, where the network structure is given. Our goal is to calculate the optimal fee assignments.
Specifically, in section \ref{sec:trees} we examine trees; trees are very natural as they connect a set of nodes with a minimal number of edges, and opening each edge costs a blockchain transaction. In addition we want all customers to prefer the PSP network, thus all paths in a given tree must cost less than $1$.
%
%
%
\section{A Linear Program for Trees}\label{sec:trees}

In this section, we find a solution to CNDF restricted to trees. We assume that every transaction makes sense in the tree, i.e., the sum of the fees on the path of every transaction is at most $1$. Therefore, by the model stated above, every user of the payment network will always use the network, and no transaction goes directly on the blockchain. 
It turns out this problem can be solved efficiently, as stated by the following result. 
\begin{theorem}
Given any tree and any transaction matrix, there exists a polynomial time algorithm to optimize the profit if every transaction can connect using the payment network.  
\end{theorem}
\begin{proof}
To solve this variation of the problem, we can use linear programming to find the optimal profit along with an optimal assignment of fees.  
In order to do so for some given tree $G = (V, E)$ and a given transaction matrix $T$, we need to first determine the objective function that we want to maximize. Moreover, we need to specify suitable inequality constraints.

We compute the objective function by analyzing how many times each transaction uses each edge in the network. This gives us an objective function 
\begin{equation*}
f(x) = \sum_{i = 0}^{m-1} c_i \cdot x_i.
\end{equation*}
The argument of the objective function, a vector $x = (x_0, \cdots , x_{m-1})$ with $m = |E|$ components represents the fees of the edges that we wish to maximize, and $c_i$ denotes the number of times the edge $i$ is used by transactions. Then, to determine the inequality constraints which are imposed by the constraint that each transaction must have a total fee of at most $1$, we define one inequality for every transaction $t$:
\begin{equation*}
\sum_{i = 0}^{m-1} e_i \cdot x_i \leq 1,
\end{equation*}
where $e_i = 1$ if edge $i$ was used for transaction $t$, and 0 otherwise. Solving this linear program (in polynomial time) finds the optimal vector $x$ of fees.
%
\hfill \qed
\end{proof}
In the following section, we remove the additional assumption that all the transactions should be facilitated by the PSP's network. The problem, now, is more complicated since the selection of transactions cannot be expressed as a linear program (but only as an ILP). Thus, we study the problem in more restricted graph structure: paths.
%
%
%
\section{Dynamic Program for Paths}\label{sec:paths}

In this section we present \emph{Algorithm \ref{Alg:dynamic}}, a polynomial-time dynamic program that achieves optimal profit in chain networks. We prove that the optimal solution has only fees that are either $0$ or $1$.
\RestyleAlgo{ruled}
\LinesNumbered
\begin{algorithm}[!t]
\caption{Dynamic Program for Paths}\label{Alg:dynamic}
\SetStartEndCondition{ }{}{}
\SetKwProg{Fn}{def}{\string:}{}
\SetKw{KwTo}{to}\SetKwFor{For}{for}{\string:}{}
\SetKwIF{If}{ElseIf}{Else}{if}{:}{elif}{else:}{}
\SetKwComment{Comment}{}{}
\SetCommentSty{} 
\AlgoDontDisplayBlockMarkers\SetAlgoNoEnd\SetAlgoNoLine

\let\oldnl\nl
\newcommand{\nonl}{\renewcommand{\nl}{\let\nl\oldnl}}

\BlankLine
\Comment{\textit{Initialization \dotfill}} \BlankLine

\nonl
$
\begin{aligned}
n &= \textrm{number of nodes}, \quad
m = \textrm{number of edges, i.e., } n-1 \\
\end{aligned}
$

\nonl Set all entries of $M[m,m,m]$ to $0$\\\setstretch{1}
\nonl Set all entries of $P[m,m]$ to $0$

\BlankLine\setstretch{1.2}
\Comment{\textit{Compute tensor M \dotfill}} 
  
		\For{every $1 \leq i \leq j \leq k \leq m$}{
    				
    		$p = 0$
    				
            
				\For{every entry T[u,v] in T}{
                	
                    
						\If{$u \le j <  v$}{
                            $p = p + T[u,v]$
						}
                        }
    				
   			$M[i,j,k] = p$
		}    				
        
\BlankLine
\Comment{\textit{Compute the dynamic programming table \dotfill}} 

	\For{every $1 \leq x \leq y \leq m$}{      	
		$P[x,y] = M[1,x,y]$ 
               
    	\For{$lastX=1$ \KwTo $x-1$}{ 
                	
        	\If{P[lastX,x-1] + M[lastX+1,x,y] \textgreater P[x,y]}{ 
                    	
            	$P[x,y] = P[lastX,x-1] + M[lastX+1,x,y]$ \\\setstretch{1}
            	Store edges with a fee of 1, i.e., $x$ and edges that have a 
                fee of $1$ for $P[lastX,x-1]$

           	}                
		}
     }
       
\BlankLine\setstretch{1}
profit = maximum entry in $P$ \\
fee assignment = edges with a fee of $1$ stored for the maximum table entry
\end{algorithm}
%
%
%
%
First we compute tensor  $M$, where $ M[i,j,k]$ is the profit from all transactions in the interval $[i,k]$ when the fee of edge $j$ is $1$ and every other fee is $0$. Then, we compute matrix $P$, the maximum entry of which is the optimal profit.
$P[lastX,x-1]$ denotes the maximum profit when setting the fee of the edge with index $lastX$ to $1$ (it is possible that more edges have a fee of $1$ before that, but $lastX$ is the last edge where this is the case) and only using the edges up to $x-1$. 
$M[lastX+1, x, y]$ denotes the profit from edges in the interval $[lastX+1,y]$ while the $x$-th edge's fee is $1$.
For some fixed $x,y$ we iterate over all possible profits of the preceding part of the graph, add it to the profit of the corresponding current interval and only consider the maximal profit (if larger that  setting the $x$-th edge's fee to 1).
%

To retrieve a fee assignment that has optimal profit, we can do the following:
We define an additional matrix $E$, where the entries of each row are initialized with the number of the row. 
Now, every time we update an entry $P[x,y]$, we set $E[x,y] = [x, E[lastX,x-1]]$. These denote the edges that are assigned a fee of 1 to attain the calculated profit.
When the algorithm has ended, we read the entry $E[x',y']$, where $x',y'$ are the indices of the maximum value in $P$, and set the fee of the edge contained in $E[x',y']$ to 1 in the optimal fee assignment.
\paragraph{Correctness and Runtime.}
We prove the correctness of Algorithm \ref{Alg:dynamic} and analyze its time complexity. 
In Algorithm \ref{Alg:dynamic}, an edge is either assigned a fee of $1$ or $0$. The following lemma states that these are indeed the only two values we need to consider. 
\begin{lemma}\label{lemma:zero_one}
For every given path and for every set of transactions, the optimal profit can always be achieved by assigning edges a fee $0$ or $1$. 
\end{lemma}
\begin{proof}
Assume that we are given some optimal fee assignment $f=(f_1,f_2,\ldots,f_m)$ on the path of length $m$, and but this assignment may use other values, not only $0$ or $1$. We show that only using $0$ and $1$ one also can reach the same (or even more) profit. 

Based on the given fee assignment $f$, we compute the set $S$ of all maximal intervals (i.e., there does not exist a pair of intervals $(i,j)$ and $(i',j')$ such that $i\le i'$ and $j\ge j'$) where the sum of the fees on the edges in that interval is less or equal to $1$. That is, an interval $(i,j)$ is in $S$ if and only if it satisfies that 
$\sum_{k=i}^j f_k \le 1$ and $\sum_{k=i-1}^j f_k >1 $ (or $i=1$) and $\sum_{k=i}^{j+1} f_k > 1$ (or $j=m$).
The optimal profit can be obtained by solving a linear program.  
It is well known that every linear program reaches its optimal at the vertex of the feasible region. 
Hence, we only need to show that every entry of every vertex of the feasible region defined above is either $0$ or $1$. 
Equivalently, we show that every feasible solution is a convex combination of vectors with only $0$ and $1$. 

We prove this by induction on the length of the path. 
For the base case, when the length is $1$, i.e., a single edge, it is trivial. 
Now assume that this result holds for paths of length smaller than $m$, and we prove that it also holds for length equals to $m$. 
The key observation is that, for any path, there always exists an assignment $f'$ with only $0$ and $1$ such that $\sum_{k=i}^j f'_k = 1$ for every $(i,j)\in S$, as follows:
\begin{enumerate}
\item Let $f'_k = 0$ for all $k$.
\item For $k$ from $1$ to $m$, consider all intervals $(i,j)$ in $S$ such that $i\le k\le j$. If all such intervals $(i,j)$ satisfy $\sum_{t=i}^j f'_t = 0$, then let $f'_k =1$. 
\end{enumerate}
We define $K:= \{k: f'_k = 1\}$ and let $\theta :=\min \{ f_k : k\in K \}$. Now we write $f$ as a convex combination $f = \theta \cdot f' + (1-\theta)\cdot f''.$ Since $\sum_{k=i}^j f'_k = 1$ for every $(i,j)\in S$, it follows that $\sum_{k=i}^j f''_k \le 1$ for every $(i,j)\in S$. By the definition of $\theta,$ we know that there exists at least one index $t$ such that $f''_t = 0$ ($f_t = \theta$). According to these two facts, $f''$ can be considered as a feasible solution for the path of length $n-1$, which by the induction hypothesis is also a convex combination of vectors with only $0$ and $1$. The lemma is proved. 
\hfill \qed
\end{proof}
%

The above lemma is useful in pruning search space, but it is still exponential ($2^m$) if we do a brute force search. Our dynamic programming method makes the search space polynomial in $m$, which is shown in the following theorem. 
%
\begin{theorem}\label{thm:correctness}
Algorithm \ref{Alg:dynamic} returns the optimal solution and the time complexity of the algorithm is $\mathcal{O}(n^5)$.
\end{theorem}
\begin{proof}
Let $OPT(x,y)$ denote the profit of the sub-path from edge 1 up to and including edge $y$ where we set the fee of edge $x$ ($x \leq y$) to 1. We claim that $OPT(x,y)$ fulfills the following recurrence:

\begin{equation*}
OPT(x,y) = \max
\begin{cases} 
M[1,x,y] & \textit{(Case 1)} \\
P[lastX,x-1] + M[lastX+1,x,y] & \textit{(Case 2)} 
\end{cases}
\end{equation*}

If $OPT(x,y)$ is equal to \textit{(Case 1)}, this means that we reach the maximum profit in the subgraph from edge $1$ to $y$ by only setting the fee of edge $x$ to $1$ in the entire subgraph. Consequently, every transaction, that only uses edges from this subgraph, can generate profit.

Otherwise, if $OPT(x,y)$ happens to be \textit{(Case 2)}, we know that there are at least two edges with a fee of $1$ in the subgraph from edge $1$ to $y$, namely on edge $x$ and on edge $lastX$. Therefore, profit is generated by transactions in the first part of the subgraph, i.e. from edge 1 to $x-1$, and at the same time in the second part, that is from $lastX+1$ to $y$. However, no transactions, which use edges in both parts of the subgraph, can generate profit, as such a transaction would then cross both edges with a fee of 1. 

Because of this, we can iterate over every possible sum of the profits of $P[lastX,x-1]$ and $M[lastX+1,x,y]$ and choose the maximum thereof. Note, that we do not necessarily choose the maximum for both terms, but instead pick the maximal sum or otherwise we might only obtain a locally optimal solution. This method can be used, since we were able to split the subgraph from 1 to $y$ in two parts as explained above. Moreover, we have already precomputed both terms: $M[lastX+1,x,y]$ was computed at the very beginning of the algorithm and $P[lastX,x-1]$ is always an entry of the table that was the result of a prior computation with exactly the same recurrence. 
	
The tensor $M$ can be computed in time $\mathcal{O}(n^5)$.
The computation of the table $P$ can be accomplished in time $\mathcal{O}(n^3)$, since we have $3$ loops that iterate over parts of the edge indices. 
Therefore, the complete algorithm can be implemented with runtime $\mathcal{O}(n^5)$.  \hfill \qed
\end{proof}	
%
%
%
\section{Payment Hub: a Near-Optimal Solution}
In this section, we present a near-optimal solution to the \emph{CNDF problem}. Please note that the optimal solution is not always a tree. For example, if we consider three nodes with many transactions between every pair, the optimal payment network is the triangle with a fee of $1$ on each edge. A tree will connect the three nodes with a two-edge path, hence none of the trees achieve maximum profit. 
We show that if the optimal network is connected, then the star graph, where the center is an additional node acting as a payment hub, is a near optimal solution.

We denote $opt(T)$ the profit, $G_{opt}$ the graph and $f_{E_{opt}}$ the fee assignment of the optimal solution for a given transaction matrix $T$. 
Moreover, we denote $S=(V_S, E_S)$ the star graph that includes all nodes $V$ and an additional one, $c$, as the center of the star. We assign uniform fees to all the edges, $f_e= 0.5, \forall e\in E_S$.
\begin{theorem}
If $G_{opt}$ is connected, then $p(S, T, f_{E_S}) \geq opt(G_{opt}, T, f_{E_{opt}}) - 1$.
\end{theorem}
\begin{proof}
If $G_{opt}$ is one connected component, then $|E_{opt}| \geq n-1$. 
For the star graph $S$, we have $V_S=V+c$ and $|E_S|= n \leq E_{opt} + 1$. Furthermore, the sum of fees on all shortest path is equal to $1$ due to the uniform fees equal to $0.5$ and the star structure.
The profit function maximizes its value when all transactions go through the graph $G_{opt}$ with total fee equal to $1$, hence 
$$ opt(G_{opt}, T, f_{E_{opt}}) \leq \sum_{i,j \in V}  T[i,j] - |E_{opt}|  \leq \sum_{i,j \in V}  T[i,j] - |E_S| + 1 = p(S, T, f_{E_S}) +1$$
The last equality holds since the sum on every shortest path equals to $1$. \hfill \qed
\end{proof}
\emph{Discussion on network connectivity.} In a monetary system we expect some nodes to be highly connected, representing big companies that transact with many nodes on the network. These highly connected nodes assist in connecting the entire network in one big connected component, as we initially assumed. 
%
%
\section{Related Work}
The Lightning Network \cite{poon2015lightning} for Bitcoin \cite{nakamoto2008bitcoin} and the Raiden Network \cite{raiden2017} for Ethereum \cite{ethereum} are the most prominent implemented decentralized path-based transaction networks for payment channels, even though similar proposals existed earlier \cite{DW2015channels}. 

Recent work has mainly focused on designing routing algorithms for these networks. The goal of these algorithms is to efficiently find a route in the network that has enough capital to facilitate the current transaction.
Prihodko et al. introduced Flare \cite{prihodko2016flare}, a routing algorithm for the Lightning network. Flare can quickly discover routes but nodes  need to collect information on the Lightning network topology. 
The IOU credit network SilentWhispers \cite{malavolta2017silentwhispers} utitizes landmark routing to discover multiple paths and then performs multiparty computation to decide how many funds to send along each path.
A more recent work by Roos et al. \cite{roos2018routing} introduces SpeedyMurmurs, which uses embedding-based path discovery to find routes from sender to receiver.
In all these algorithms, the task is to find a route to facilitate a customer's transaction through the payment network. 
Routing very much is orthogonal to our goal of finding the right fees.

Heilman et al. \cite{heilman2017tumblebit} propose a Bitcoin-compatible construction of a payment hub for fast and anonymous off-chain transactions through an untrusted intermediary.
Green et al. present Bolt \cite{green2017bolt} (Blind Off-chain Lightweight Transactions) for constructing privacy-preserving unlinkable and fast payment channels. Both protocols focus on constructing anonymous and private systems that can act as payment hubs. We show that constructing a payment hub is a near optimal strategy with respect to a PSP's profit.
%
%
\section{Conclusion} 
To the best of our knowledge, we are the first to introduce a framework for network design with fees on payment channels. We present algorithms that calculate the optimal fee assignments given a path or a tree as a graph structure. 
Furthermore, we prove the star is a near-optimal solution when we allow adding an extra node to act as an intermediary for the customers. This implies that the construction of payment hubs is an almost optimal strategy for a PSP. 

\newpage

\bibliographystyle{splncs04}
\bibliography{references}

\end{document}